\documentclass[runningheads,envcountsame]{llncs}
 \usepackage[latin1]{inputenc}
\usepackage[utf8]{inputenx} 
\usepackage{a4wide} 
\usepackage[T1]{fontenc} 
\usepackage{hyperref}
\usepackage{enumerate} 
\usepackage{xcolor} 
\usepackage{soul} 
\usepackage{amsmath} 
\usepackage{nccmath} 
\usepackage{amssymb} 
\usepackage{stmaryrd} 
\usepackage{fancyhdr} 
\usepackage{enumitem} 
\usepackage{titleref} 
\usepackage{mathrsfs} 
\usepackage{cancel}

\usepackage{color} 
\usepackage{longtable}

\DeclareMathOperator{\id}{Id}
\DeclareMathOperator{\Aut}{Aut}
\DeclareMathOperator{\Pol}{Pol}
\DeclareMathOperator{\Csp}{CSP}
\DeclareMathOperator{\Nsp}{NSP}
\DeclareMathOperator{\Forb}{Forb}
\DeclareMathOperator{\pol}{Pol}
\DeclareMathOperator{\proj}{Proj}

\usepackage{cite}

\title{Hardness of Network Satisfaction for Relation Algebras with Normal
	Representations}

\author{Manuel Bodirsky\thanks{The author has received funding from the European Research Council (Grant Agreement no. 681988, CSP-Infinity)} 
\and Simon Kn\"auer \thanks{The author is supported by DFG Graduiertenkolleg 1763 (QuantLA).}}

\institute{Institut f\"{u}r Algebra, TU Dresden, 01062 Dresden, Germany}

\begin{document}

\setcounter{page}{1}

\maketitle
\begin{abstract}
We study the computational complexity 
of the general network satisfaction problem for
a finite relation algebra $A$ with a normal representation $B$.
If $B$ contains a non-trivial equivalence relation with a finite number of equivalence classes, then the network satisfaction problem for $A$ is NP-hard. As a second result, we prove hardness if $B$ has domain size at least three and contains no
non-trivial equivalence relations but a symmetric
atom $a$ with a forbidden triple $(a,a,a)$, that is, 
$a \not\leq a \circ a$. We illustrate how to apply 
our conditions on two small relation algebras. 
\end{abstract}

\section{Introduction}
\par 
Many computational problems in temporal and spatial reasoning can be formulated as network
satisfaction problems for a fixed finite relation algebra~\cite{Duentsch,NebelRenzSurvey,Qualitative-Survey}. Famous examples of finite relation algebras that have been studied in this context are the Point Algebra, the Left Linear Point Algebra, 
Allen's Interval Algebra, RCC5, and RCC8, just to name a few; much more material about relation algebras can be found in~\cite{HirschHodkinson}. Robin Hirsch~\cite{Hirsch} asked in 1996 the \emph{Really Big Complexity Problem (RBCP)}: can we classify the computational complexity of the network satisfaction problem for every finite relation algebra? For example, the network satisfaction problem for the Point Algebra and the Left Linear Point Algebra are polynomial-time tractable~\cite{PointAlgebra,BodirskyKutzAI}, while it is NP-complete for the other relation algebras mentioned above~\cite{Allen,NebelRenz}. A finite relation algebra with an undecidable network satisfaction problem has been found by Hirsch~\cite{Hirsch-Undecidable}.

An important notion in the theory of representability of finite relation algebras are \emph{normal representations}, i.e., representations that are fully universal, square, and homogeneous~\cite{Hirsch}. The network satisfaction problem
for a relation algebra with a normal representation can be seen as the constraint satisfaction problem
for an infinite structure $\mathfrak B$ that is homogeneous and finitely bounded (these concepts from model theory will be introduced in Section~\ref{sect:normal}). The network satisfaction problem is in this case in NP and a complexity dichotomy has been conjectured~\cite{BPP-projective-homomorphisms}.
There is even a promising candidate condition for the boundary between NP-completeness and containment in P; the condition can be phrased in several equivalent ways~\cite{BKOPP,BodirskyRamics}. 
However, this conjecture has not yet been verified 
for the homogeneous finitely bounded structures
that arise as the normal representation 
of a finite relation algebra.

We present
some first steps towards a solution to the RBCP for relation algebras $\bf A$ with a normal representation $\mathfrak B$. Our approach
is to study the automorphism group $\Aut({\mathfrak B})$ of $\mathfrak B$ and to identify properties 
that imply hardness. 
Because of the homogeneity of $\mathfrak B$, 
one can translate
back and forth between properties of $\bf A$ 
and properties of $\Aut({\mathfrak B})$. 
For example, $\Aut({\mathfrak B})$
is primitive if and only if $\bf A$ contains no
 equivalence relation which is different from
 the \emph{trivial} equivalence relations 
 $\id$ and $1$. 
Specifically, we show that the network
satisfaction problem for $\bf A$ is NP-complete if 
\begin{itemize}
\item $\Aut({\mathfrak B})$
is primitive, $|B|>2$ and $\bf A$ has a symmetric
atom $a$ with a forbidden triple $(a,a,a)$, that is, 
$a \not\leq a \circ a$ (Section~\ref{sect:fin});
\item $\Aut({\mathfrak B})$ has a congruence 
with at least two but finitely many equivalence classes (Section~\ref{sect:primitive}). 
\end{itemize}

In our proof we use the so-called \emph{universal-algebraic approach} which has recently led to a full classification of the computational complexity of 
constraint satisfaction problems for $\mathfrak B$ if the domain of $\mathfrak B$ is \emph{finite}~\cite{BulatovFVConjecture,ZhukFVConjecture}. The central insight is
that the complexity of the CSP is for finite $\mathfrak B$
fully determined by the \emph{polymorphism clone}
$\Pol(\mathfrak B)$ of $\mathfrak B$. This result extends to 
homogeneous structures with finite relational signature (more generally, to $\omega$-categorical structures~\cite{BodirskyNesetrilJLC}). 
Both of our hardness proofs come from the technique of factoring $\Pol({\mathfrak B})$ with respect
to a congruence with finitely many classes,
and using known hardness conditions from corresponding finite-domain constraint satisfaction problems. 
The article is fully self-contained: we introduce the network satisfaction problem (Section~\ref{sect:nsp}), normal representations (Section~\ref{sect:normal}), 
and the universal algebraic approach (Section~\ref{sect:ua}). 


\section{The (General) Network Satisfaction Problem}
\label{sect:nsp}
Network satisfaction problems have been introduced in~\cite{LadkinMaddux}, capturing
well-known computational problems, e.g., for Allen's Interval Algebra~\cite{Allen}; see~\cite{Duentsch} for a survey. 
An \emph{algebra} in the sense of universal algebra is a set together with operations on this set, each equipped with an arity $n \in {\mathbb N}$. In this context, operations of arity zero are viewed as constants. The \emph{type} of an algebra is a tuple that represents  the arities of the operations. 
For the definitions concerning relation algebras, we  basically follow \cite{Maddux2006}. 

\begin{definition}
	Let $D$ be a set and $E \subseteq D^2$ an equivalence relation. Let  $(\mathcal{P}(E); \cup,   \bar{} , 0,1,\id,  ^\smile, \circ ) $ be an algebra of type $(2,1,0,0,0,1,2)$ with the following operations:\begin{enumerate}
		\item $A \cup B :=  \{(x,y)  \mid (x,y)\in A \textup{~or~} (x,y)\in B \} $,
		\item $\bar{A}:= E \setminus A $,
		\item $0:= \emptyset$,
		\item $1:= E$,
		\item $\id:= \{(x,x) \mid x\in D \} $,
		\item $A^\smile := \{(x,y) \mid (y,x)\in A  \}$,
		\item $A\circ B := \{(x,z) \mid \exists y \in D: (x,y)\in A  \textup{~ and~ } (y,z)\in B \}$.
	\end{enumerate}
	A subalgebra of $(\mathcal{P}(E); \cup,   \bar{} , 0,1,\id,  ^\smile, \circ)$ 
	is called a \emph{proper relation algebra}.
\end{definition}

A \emph{representable relation algebra} is an algebra of type $(2,1,0,0,0,1,2)$ that is isomorphic (as an algebra) to a proper relation algebra.
We denote algebras by bold letters, like $\mathbf{A}$; the underlying domain of an algebra $\mathbf{A}$ is denoted with the regular letter $A$. An algebra
$\bf A$ is finite if $A$ is finite. 
We do not need the more general definition of an \emph{(abstract) relation algebra} (for a definition see for example \cite{Maddux2006}) because the
	network satisfaction problem for
	relation algebras that are not 
	representable is trivial. 
We use the language of model theory to define \emph{representations} of relation algebras;
the definition is essentially the same as the one given in~\cite{Maddux2006}. 


\begin{definition}
	A relational structure $\mathfrak{B}$  is called a \emph{representation} of a relation algebra $\mathbf{A}$ if 
	\begin{itemize}
	\item $\mathfrak{B}$ is an $A$-structure
		with domain $B$ (i.e., each element $a \in A$ is used as a relation symbol denoting a binary relation $a^{\mathfrak B}$ on $B$);
		\item there exists an equivalence relation $E \subseteq B^2$ such that the set of relations of $\mathfrak B$ is the domain of a subalgebra of
		$(\mathcal{P}(E); \cup,   \bar{} , 0,1,\id,  ^\smile, \circ )$; 
		\item the map that sends $a \in A$
		to $a^{\mathfrak B}$ is an isomorphism between
		$\bf A$ and this subalgebra. 
	\end{itemize}
\end{definition}

\begin{remark}
	For a relation algebra $\mathbf{A}=(A; \cup,   \bar{} , 0,1,\id,  ^\smile, \circ )$ the algebra $(A; \cup,   \bar{} , 0,1) $ is a Boolean algebra. 
	With respect to this algebra there is a partial ordering on the elements of a relation algebra. We denote this with $\subseteq$ since in proper relation algebras this ordering is with respect to set inclusion. 
	The minimal non-empty relations with respect to $\subseteq$ are called the \emph{atomic relations} or \emph{atoms}; we denote the set of atoms of $\bf A$ by $A_0$. 
\end{remark}


\begin{definition}
	Let $\mathbf{A}$ be a relation algebra. An $\mathbf{A}$-network $(V;f)$  is a finite set of nodes $V$ together with a function $f\colon V \times V \rightarrow A$. 
	
	Let $\mathfrak{B}$ be a representation of $\mathbf{A}$. An $\mathbf{A}$-network $(V;f)$  is \emph{satisfiable in $\mathfrak{B}$} if there exists an assignment $s\colon V \rightarrow B$ such that for all $x,y \in V$
	$$(s(x),s(y)) \in f(x,y) ^\mathfrak{B}.$$
	An $\mathbf{A}$-network $(V;f)$  is \emph{satisfiable} if there exists some representation $\mathfrak{B}$ of $\mathbf{A}$ such that  $(V;f)$  is satisfiable in $\mathfrak{B}$.
\end{definition}


\begin{definition}
	The \emph{(general) network satisfaction problem} for a finite relation algebra $\mathbf{A}$,
	denoted by $\Nsp(\bf A)$, is the problem of deciding whether a given $\mathbf{A}$-network is satisfiable. 
\end{definition}

\section{Normal Representations and CSPs}
\label{sect:normal}
We recall a connection between network satisfaction problems and constraint satisfaction problems that is presented in more detail  in~\cite{Qualitative-Survey,BodirskyRamics}. 

\begin{definition}[from \cite{Hirsch}]
	Let $\mathbf{A}$ be a relation algebra. An $\mathbf{A}$-network $(V;f)$  is called \emph{atomic} if the image of $f$ only contains atoms and if $$f(a,c)\subseteq f(a,b) \circ f(b,c).$$
	\end{definition}

The last line ensures a ``local consistency'' of the atomic $\mathbf{A}$-network with respect to the multiplication rules in the relation algebra $\mathbf{A}$. This property is in the literature sometimes called ``closedness'' of an $\mathbf{A}$-network \cite{HirschAlgebraicLogic}.

\begin{definition}[from \cite{Hirsch}]
	A representation $\mathfrak{B}$ of a relation algebra $\mathbf{A}$ is called 
	\begin{itemize}
		\item \emph{fully universal} if every atomic $\mathbf{A}$-network is satisfiable in $\mathfrak{B}$;
		\item \emph{square} if $1^\mathfrak{B}=B^2$;
		\item  \emph{homogeneous} if every isomorphism of finite substructures of $\mathfrak{B}$ can be extended to an automorphism;
		\item	\emph{normal} if it is fully universal, square and homogeneous.
	\end{itemize}

\end{definition}

If a relation algebra $\mathbf{A}$ has a normal representation $\mathfrak{B}$ then the problem 
of deciding whether an $\mathbf{A}$-network is satisfiable in \emph{some} representation reduces to a question whether it is satisfiable in the concrete representation $\mathfrak{B}$. Such decision problems are known as constraint satisfaction problems, which are formally defined in the following.

\begin{definition}
	Let $\mathfrak{B} $ be a $\tau$-structure for a finite relational signature $\tau$. The \emph{constraint satisfaction problem} of $\mathfrak{B}$ is the problem of deciding for a given finite $\tau$-structure $\mathfrak{C}$ whether there exists a homomorphism from $\mathfrak{C}$ to $\mathfrak{B}$.
\end{definition}

To formulate the connection between NSPs and CSPs, we have to give a translation between networks and structures. On the one hand we may view an $\mathbf{A}$-network $(V;f)$ as an $A$-structure $\mathfrak C$ with domain $C := V$ where 
$(a,b) \in f(a,b)^{\mathfrak C}$. 
 On the other hand we can transform an $A$-structure $\mathfrak{C}$ into an $\mathbf{A}$-network $(V;f)$ with $V = C$ and by defining the network function $f(x,y)$ for $x,y \in C$ as follows: 
let $X$ be the set of all relations that hold on $(x,y)$ in $\mathfrak{C}$. If $X$ is non-empty we define $f(x,y):= \bigcup X$; otherwise $f(x,y):= 1$.

\begin{proposition}[see~\cite{BodirskyRamics}]
	Let $\mathfrak{B}$ be a normal representation of a finite relation algebra $\mathbf{A}$.
	Then $\Nsp(\mathbf{A})$ and $\Csp(\mathfrak{B})$ are the same problem (up to the translation showed above).
\end{proposition}

The following is an important notion in model theory and the study of infinite-domain CSPs.
Let $\mathcal F$ be a finite set of finite $\tau$-structures. Then $\Forb(\mathcal F)$ is the class of all finite $\tau$-structures that embed no $\mathfrak{C}\in \mathcal F$.
A  class $\mathcal{C}$ of finite $\tau$-structures is called \emph{finitely bounded} if $\mathcal{C}=\textup{Forb}({\mathcal F})$ for a finite set ${\mathcal F}$. 
A structure $\mathfrak{B}$ is called \emph{finitely bounded} if the class of finite structures that embed into $\mathfrak{B}$ is finitely bounded.

\begin{proposition}[see~\cite{BodirskyRamics}]
	Let $\bf A$ be a finite relation algebra with
	a normal representation $\mathfrak B$. 
	Then $\mathfrak{B}$ is finitely bounded and $\Csp(\mathfrak{B})$ and $\Nsp(\mathbf{A})$ are in NP.
\end{proposition}

\section{The Universal Algebraic Approach}
\label{sect:ua}
This section gives a short overview of the important notions and concepts for the universal-algebraic approach to the computational complexity of CSPs. 

\subsection{Clones}
We start with the definition of an operation clone. 

\begin{definition}
	Let $B$ be some set. Then $\mathcal{O}^{(n)}_B$ denotes the set of $n$-ary operations on $B$ and $\mathcal{O}_B:= \bigcup_{n\in \mathbb{N}} \mathcal{O}^{(n)}_B$.
	A set $\mathscr{C} \subseteq \mathcal{O}_B $ is called a \emph{operation clone (on $B$)} if it contains all projections and is closed under composition, that is, for every $ f\in\mathscr{C} \cap \mathcal{O}^{(k)}_B $ and all $g_1,\ldots, g_k \in \mathscr{C}\cap \mathcal{O}^{(n)}_B $ the $n$-ary operation $f(g_1,\ldots,g_k)$ with
	$$f(g_1,\ldots,g_k)(x_1,\ldots, x_n) :=  f(g_1(x_1,\ldots, x_n)     ,\ldots,g_k(x_1,\ldots, x_n) )     $$
	is also in $\mathscr{C}$. 
	We denote  the $k$-ary operations of $\mathscr{C}$ by $\mathscr{C}^{(k)}$. 
\end{definition}

\begin{definition}
	Let $\mathfrak{B}$ be a relational structure.
	Then $f$ \emph{preserves} a relation $R$ of $\mathfrak{B}$ if the component-wise application of $f$ on tuples $r_1, \ldots, r_k \in R$ results in a tuple of the relation. If $f$ preserves all relations of $\mathfrak B$ then $f$ is called a \emph{polymorphism} of $\mathfrak B$. 
	The set of all polymorphisms of  arity $k \in \mathbb{N}$ is denoted by $\pol^{(k)}(\mathfrak{B})$ and $\pol(\mathfrak{B}) := \bigcup_{k \in {\mathbb N}} \pol^{(k)}(\mathfrak{B})$ is called the \emph{polymorphism clone of $\mathfrak{B}$}.
\end{definition}

Polymorphisms are closed under the composition and a projection is always a polymorphism, therefore a polymorphism clone is indeed an operation clone.

\begin{definition}
Let $\mathscr{C}$ and $\mathscr{D}$ be operation clones.  
A function $\mu \colon \mathscr{C} \rightarrow \mathscr{D}$ is called \emph{minor-preserving}
if it maps every operation to an operation of the same arity and 
satisfies for every $ f\in\mathscr{C}^{(k)}$  and all projections $p_1,\ldots, p_k $ the following identity:
$$\mu(f(p_1,\ldots,p_k))= \mu(f)(p_1,\ldots,p_k).$$
\end{definition}

Operation clones $\mathscr C$ on countable sets $B$ can be equipped with the following complete ultrametric $d$. Assume that $B=\mathbb{N}$.  For two polymorphisms $f$ and $g$ of different arity we define $d(f,g) =1$. If $f$ and $g$ are both of arity $k$ we have
$$d(f,g) :=2^{- \min\{n\in \mathbb{N} \mid \exists s \in \{1,\ldots, n \}^k : f(s)\not = g(s) \}}.$$

The following is a straightforward consequence of the definition. 

\begin{lemma}\label{u.c.}
	Let $\mathscr{D}$ be an operation clone on $B$ and $\mathscr{C}$ an operation clone on $C$ and let $\nu \colon \mathscr{D}\rightarrow \mathscr{C}$ a map. Then $\nu$ is uniformly continuous (u.c.) if and only if
	$$\forall n \geq 1 \; \exists \textup{~finite~} F\subset D \forall f,g \in  \mathscr{D}^{(n)} : f|_F =g|_F \Rightarrow \nu(f)=\nu(g).$$
\end{lemma}

In order to demonstrate the use of polymorphisms in the study of CSPs we have to define primitive positive formulas. Let $\tau$ be a relational signature. A first-order formula $\varphi(x_1,\ldots,x_n)$ is called \emph{primitive positive} if it has the  form
$$ \exists x_{n+1},\ldots, x_{m} (\varphi_1\wedge \cdots \wedge \varphi_s)$$
where $\varphi_1, \ldots, \varphi_s$ are atomic formulas, i.e., formulas of the form $R(y_1,\ldots, y_l)$ for $R\in \tau$ and $y_i \in \{x_1,\ldots, x_m\}$, of the form $y=y'$ for $y,y' \in \{x_1,\ldots x_m\}$, or of the form \emph{false} and \emph{true}. 
We have the following correspondence between polymorphisms and primitive positive formulas (or relations that are defined by them). Note that all of the statements in the following hold in a more general setting, but we only state them here for normal representations of finite relation algebras. 

\begin{theorem}[follows from~\cite{BodirskyNesetrilJLC}]\label{polinv}
	Let $\mathfrak{B}$ be a normal representation of a finite relation algebra $\mathbf{A}$. Then the set of primitive positive definable relations in $\mathfrak{B}$ is exactly the set of relations that are preserved by $\pol(\mathfrak{B})$.
\end{theorem}

A special type of polymorphism plays an important role in our analysis.
\begin{definition}
	Let $f$ be an $n$-ary operation on a countable set $X$. Then $f$ is called cyclic if 
	$$\forall x_1,\ldots x_n \in X : f(x_1,\ldots,x_n)= f(x_n,x_1 \ldots, x_{n-1}). $$
\end{definition}
%

We write $\proj$ for the operation clone on a two-element set that consists of only the projections. 

\begin{theorem}[from~\cite{Cyclic,wonderland}]\label{existence cyclic}
	Let $\mathscr{C}$ be an operation clone on a finite set $C$. If there exists no minor-preserving map $\mathscr{C} \rightarrow \proj$ then $\mathscr{C}$ contains for every prime $p >|C|$ a $p$-ary cyclic operation. 
\end{theorem}

Note that every map between operation clones on finite domains is uniformly continuous.



\begin{theorem}[from~\cite{wonderland}]\label{hardness infinite ucmp}
	Let $\mathfrak{B}$ be normal representation of a finite relation algebra. If there is a uniformly continuous minor-preserving map $\pol(\mathfrak{B}) \rightarrow \proj$, then $\Csp(\mathfrak{B})$ is NP-complete. 
\end{theorem}

\subsection{Canonical Functions}

Let $\mathfrak{B}$ be a normal representation of a finite relation algebra $\mathbf{A}$. 

\begin{definition}
	Let $a_1, \ldots, a_k \in A$. Then $(a_1, \ldots , a_k)^{\mathfrak B}$ denotes a binary relation on $B^k$ such that for $x,y \in B^k$ 
	$$(a_1, \ldots , a_k)^{\mathfrak B}(x,y)  :\Leftrightarrow  \bigwedge_{i\in \{1,\ldots, k\} }  a_i^\mathfrak{B}(x_i,y_i) .$$
\end{definition}

Recall that $A_0$ denotes the set of atoms of a  representable relation algebra $\mathbf{A}$.

\begin{definition}
	Let $x,y \in B^k$. Since $\mathfrak{B}$ is square 
	there are unique $a_1, \ldots, a_k \in A_0$
	such that $(a_1,\ldots,a_k)^{\mathfrak B}(x,y)$. 
	Then we call $(a_1,\ldots,a_k)^{\mathfrak B}$ 
	the \emph{configuration of $(x,y)$}. 
	If $a_1, \ldots, a_k\in X\subseteq A_0$ then $(a_1,\ldots,a_k)$ is called an \emph{$X$-configuration}.
\end{definition}

We specialise the concept of \emph{canonical functions} (see, e.g.,~\cite{BodPin-CanonicalFunctions}) to our setting. 

\begin{definition}
	Let $f$ be a $k$-ary operation on $B$.
	Let $X \subseteq A_0$ and let $T$ be the set of all $X$-configurations. Then $f$ is called \emph{$X$-canonical} if there exists a map $\overline{f}\colon T \rightarrow A_0$ such that for 
	every $(a_1,\dots,a_k) \in T$ and $(x,y) \in (a_1, \ldots, a_k)^{\mathfrak B}$ we have $(f(x),f(y)) \in \big (\overline{f}(a_1, \ldots, a_k)\big)^\mathfrak{B}$.
	If $X=A_0$ then $f$ is called \emph{canonical}.
\end{definition}



An operation $f \colon B^n \to B$ is called \emph{conservative} 
if for all $x_1,\dots,x_n \in B$ 
$$f(x_1,\dots,x_n) \in \{x_1,\dots,x_n\}.$$
If $\mathfrak B$ is a finite structure such that every
polymorphism of $\mathfrak B$ is conservative, then
$\Csp(\mathfrak B)$ has been classified already before
the proof of the Feder-Vardi conjecture, and there are several proofs~\cite{Conservative,Bulatov-Conservative-Revisited,Barto-Conservative}.
The polymorphisms of normal representations of finite relation algebras satisfy a strong property that resembles conservativity. 

\begin{proposition}
	Let $\mathfrak{B}$ be a normal representation. Then every $f \in \Pol^{(n)}$ 
	is \emph{edge-conservative}, that is,  for all $x,y \in B^n$ with configuration 
	$(a_1, \ldots, a_n)^{\mathfrak B}$ it holds that
	$$(f(x), f(y)) \in \left (\bigcup_{i\in \{1,\ldots,n \}} a_i \right)^\mathfrak{B}.$$
\end{proposition}
\begin{proof}
	By definition, $b := \bigcup_{i\in \{1,\ldots,n \}} a_i $ is part of the signature of $\mathfrak{B}$. 
	Moreover, for every $i \in \{1,\ldots, n\}$ we have that $(x_i,y_i) \in b^\mathfrak{B}$ by the assumption on the configuration of $x$ and $y$. 
	Then $(f(x),f(y)) \in b^{\mathfrak{B}}$ because $f$ preserves $b^{\mathfrak B}$. 
\qed\end{proof}

\section{Finitely Many Equivalence Classes}
\label{sect:fin}
In the following, $\bf A$ denotes a finite
relation algebra with a normal representation $\mathfrak{B}$. 


\begin{theorem}\label{hardness:finitely many classes}
	Suppose that $e \in A$ is such that $e^\mathfrak{B}$ is a non-trivial equivalence relation with finitely many classes.
	Then $\Csp(\mathfrak{B})$ is NP-complete.
\end{theorem}

\begin{proof}
	We use the notation 
	$n:= 1\setminus e$. 
	Let $\{c_1, \ldots, c_m\}$ be a set of representatives of the equivalence classes of $e^\mathfrak{B}$. We denote the equivalence class of $c_i$ by $\overline{c_i}$. 
A $k$-ary polymorphism $f \in \pol(\mathfrak{B})$ induces an operation $\overline{f}$ of arity $k$ on $C=\{\overline{c_1},\ldots,\overline{c_m} \}$ in the following way:
	$$\overline{f}(\overline{d_1}, \ldots \overline{d_k}):=\overline{ f(d_1, \ldots d_k) }$$
		for all $\overline{d_1}, \ldots \overline{d_k} \in \{\overline{c_1}, \ldots \overline{c_m}\}$. This definition is independent from the choice of the representatives since the polymorphisms preserve the relation $e^\mathfrak{B}$.  We denote the set of all operations that are induced in this way by operations from  $\pol(\mathfrak{B})$ by $\mathscr{C}$. It is easy to see that $\mathscr{C}$ is an operation clone on a finite set. Moreover, the mapping $\mu \colon  \pol(\mathfrak{B}) \rightarrow \mathscr{C}$ defined by $\mu(f):= \overline{f}$ is a minor-preserving map. To show that $\mu$ is uniformly continuous, we use Lemma~\ref{u.c.}; it suffices to observe that if two $k$-ary operations $f,g \in \Pol(\mathfrak{B})$ are equal on $F := \{c_1, \ldots, c_m\}$, then they induce the same operation on the equivalence classes.

\medskip 

	Suppose for contradiction that $\mathscr{C}$ contains a $p$-ary cyclic operation 
	for every prime $p>m$. 
	\medskip 
	
	
\underline{Case 1:}  $m=2$. 
By assumption there exists a ternary cyclic operation $\overline{f} \in \mathscr{C}$. 
Since $e^{{\mathfrak B}}$ is non-trivial,
one of the equivalence classes of $e^{\mathfrak B}$ must have size at least two. 
So we may without loss of generality assume that $\overline{c_1}$ contains at least two elements. 
Let $c_1' \in \overline{c_1}$ with $c_1 \not = c_1'$. 
	We have that $\overline{ f(c_{1}, c_{1}, c_2 ) } =\overline{ f(c_{2}, c_{1},c_1) }$ 
	which means that
	\begin{align}
	\big( f(c_{1}, c_{1}, c_2 ) ,f(c_{2},c_1,  c_1) \big) \in e^\mathfrak{B}. \label{eq:one}
	\end{align}
	 On the other hand $(n,\id,n)^{\mathfrak B} \big( (c_{1}, c_{1}, c_2 ),  (c_{2}, c_{1},c_1)\big)$.
Since $f$ is an edge conservative polymorphism we have that 
\begin{align}
\big(f(c_{1}, c_{1}, c_2 ),f(c_{2}, c_{1},c_1) \big) \in (n \cup \id)^{\mathfrak B}. \label{eq:two}
\end{align}
Combining (\ref{eq:one}) and (\ref{eq:two})
we obtain that 
\begin{align}
f(c_{1}, c_{1}, c_2 ) = f(c_{2},c_1,  c_1) . \label{eq:combine} 
\end{align}

	Similarly, $\overline{ f(c_{2}, c_{1}, c_1 ) } =\overline{ f(c_{1}, c_{2},c_1) }$. Since $f$ preserves the equivalence relation $e^{\mathfrak B}$ we also have $ \big(f(c_{1}, c_{2},c_1) , f(c'_{1}, c_{2},c_1) \big) \in  e^\mathfrak{B}$.
	But then $ (f(c_{2}, c_{1},c_1) , f(c'_{1}, c_{2},c_1) ) \in  e^\mathfrak{B}$  holds.
	Also note that $(n,n,\id)^{\mathfrak B} \big( (c_{2}, c_{1}, c_1 ),  (c'_{1}, c_{2},c_1)  \big)$ implies
	that $\big(f(c_{2}, c_{1}, c_1 ),f(c'_{1}, c_{2},c_1)\big) \in (n \cup \id)^{\mathfrak B}$. 
	These two facts together imply $f(c_{2}, c_{1},c_1) =f(c'_{1}, c_{2},c_1)$. 
	By (\ref{eq:combine}) 
	and the transitivity of equality we get 
	$f(c_{1}, c_{1},c_2) = f(c'_{1}, c_{2},c_1)$.   
	But this is impossible because 
	$(e,n,n)^{\mathfrak B} \big( (c_{1}, c_{1}, c_2 ),  (c'_{1}, c_{2},c_1)\big)$
	implies that $f(c_{1}, c_{1},c_2) \neq f(c'_{1}, c_{2},c_1)$.
	
	\medskip 
	\underline{Case 2:}  $m>2$. 	Let $f$ be a $p$-ary cyclic operation for some prime $p>m$. Consider the representatives $c_1, c_2 $ and $c_3$.
	By the cyclicity of $\overline f$
	we have 
	$$\overline{ f(c_{1}, c_{2},\ldots, c_1,c_2 ,c_3 ) } =\overline{ f(c_{3}, c_{1},c_2\ldots,  c_1,c_2) }$$
	and therefore  
			\begin{align}
\big (f(c_{1}, c_{2},\ldots, c_1,c_2 ,c_3 ) , f(c_{3}, c_{1},c_2\ldots,  c_1,c_2) \big ) \in e^\mathfrak{B}.
	\label{eq:equal}
	\end{align}
 On the other hand, 
	 $$(n,n,n, \ldots, n,n)^{\mathfrak B} \big ((c_{1}, c_{2},\ldots, c_1,c_2 ,c_3 ) , (c_{3}, c_{1},c_2\ldots,  c_1,c_2) \big ) $$
	and since $f$ preserves $n^{\mathfrak B}$ we get that 
	$$\big (f(c_{1}, c_{2},\ldots, c_1,c_2 ,c_3 ) , f(c_{3}, c_{1},c_2\ldots,  c_1,c_2) \big) \in n^\mathfrak{B},$$ contradicting (\ref{eq:equal}). 
	
	%
	%

	We showed that there exists a prime $p>m$ such  that $\mathscr{C}$ does not contain a $p$-ary cyclic polymorphism and therefore Theorem \ref{existence cyclic} implies the existence of a (uniformly continuous) minor-preserving map $\nu \colon \mathscr{C} \rightarrow \proj$. Since the composition of uniformly continuous minor-preserving maps is again uniformly continuous and minor-preserving,
	there exists a uniformly continuous minor-preserving
	map $\nu \circ \mu \colon \pol(\mathfrak{B}) \rightarrow \proj$.
	This map implies the NP-hardness of $\Csp(\mathfrak{B})$ by Theorem \ref{hardness infinite ucmp}. \qed\end{proof}

\section{No Non-Trivial Equivalence Relations}
\label{sect:primitive}
In this section $\bf A$ denotes
a finite relation algebra with a normal representation $\mathfrak B$ with $|B|>2$. 

\begin{definition} 
	The automorphism group  $\Aut({\mathfrak C})$ of a relational structure $\mathfrak C$ is called \emph{primitive} if $\Aut({\mathfrak C})$ does not preserve a non-trivial equivalence relation, i.e., the only equivalence relations that are preserved by $\Aut({\mathfrak C})$ are $\id$ and $C^2$.

\end{definition}

\begin{proposition}\label{prop:transitive} Let $a$ be an atom of $\bf A$. 
If $\Aut({\mathfrak B})$ is primitive then $a \subseteq \id$ implies  $a=\id$.
\end{proposition}
\begin{proof}
	
If  $a \subsetneq\id$ then 
$$c := \id \cup (a \circ 1 \circ a)$$


would be such that $c^{\mathfrak B}$ is a non-trivial
equivalence relation.  
\qed\end{proof}

\begin{proposition}\label{prop: kein k2} Let $a$ be a symmetric atom of $\bf A$ with $a \cap \id = 0$. 
	If $\Aut({\mathfrak B})$ is primitive then 	$a^{\mathfrak B} \circ a^{\mathfrak B} \not = \id$.
	
\end{proposition}
\begin{proof}Assume for contradiction 	$a^{\mathfrak B} \circ a^{\mathfrak B}  = \id^\mathfrak{B} $.
	This implies 	$(\id \cup  a)^{\mathfrak B} \circ (\id \cup a)^{\mathfrak B} \subset (\id \cup a)^\mathfrak{B} $ and therefore $(\id \cup  a)^{\mathfrak B}$ is an equivalence relation. Since $\mathfrak{B}$ is primitive $(\id \cup  a)^{\mathfrak B}=B^2$. By assumption $B$ contains at least $3$ elements. These elements are now all connected by the atomic relation $ a^{\mathfrak B}$. This is a contradiction to our assumption 	$a^{\mathfrak B} \circ a^{\mathfrak B}  = \id^\mathfrak{B} $. \qed	\end{proof}
Higman's lemma states that a permutation group
$G$ 
on a set $B$  
is primitive if and only if for every two distinct elements $x,y \in B$ the undirected graph with vertex set
$B$ and edge set $\big \{\{\alpha(x),\alpha(y)\} \mid \alpha \in G \big \}$ is connected (see, e.g.,~\cite{CameronPermutationGroups}). 
We need the following variant of this
result for $\Aut({\mathfrak B})$; we also present its proof since we are unaware of any reference in the literature. 
If $a \in A$ then a sequence $(b_0,\dots,b_n) \in B^{n+1}$
is called an \emph{$a$-walk (of length $n$)} if $(b_i,b_{i+1}) \in a^{\mathfrak B}$ for every $i \in \{0,\dots,n-1\}$ (we count the number of traversed edges rather than the number of vertices when defining the length). 


\begin{lemma}\label{pathexistence}
	Let $a \in A$ be a symmetric atom of $\bf A$ with $a \cap \id = 0$
	and suppose that $\Aut(\mathfrak B)$ is primitive. Then there exists an $a^\mathfrak{B}$-walk of
	even length between any $x,y \in B$. 
	Moreover, there exists 
	$k\in \mathbb{N}$ such that for all $x,y \in B$ there exists an $a^\mathfrak{B}$-walk of length $2k$ between $x$ and $y$.
\end{lemma}

\begin{proof}If $R$ is a binary relation then 
$R^k= R \circ R \circ \cdots \circ R$ denotes the $k$-th relational power of $R$. 
The sequence of binary relations $L_n :=\id^\mathfrak{B} \cup  \bigcup_{k=1}^n (a^\mathfrak{B})^{2k}$ is non-decreasing by definition and terminates because all binary relations are unions of at most finitely many atoms. Therefore, there exists $k \in \mathbb{N}$ such for all $n \geq k$ we have $L_n=L_k$.  Note that $L_k$ is an equivalence relation, namely the relation ``there exists an $a^\mathfrak{B}$-walk of even length between $x$ and $y$''. Since $\mathfrak{B} $ is primitive $L_k$ must be trivial. 
If $L_k = B^2$ then there exists an 
$a^\mathfrak{B}$-walk of length $2k$ between any two $x,y \in B$ and we are done. 
Otherwise, $$L_k = \{(x,x) \mid x \in B\} = \id^{\mathfrak B}.$$ 
Since $a$ is symmetric $a^{\mathfrak B} \circ a^{\mathfrak B} \not = 0$ and $a^{\mathfrak B} \circ a^{\mathfrak B}$ contains therefore an atom.
But then $a^{\mathfrak B} \circ a^{\mathfrak B}  \subseteq L_k$ implies by
 Proposition~\ref{prop:transitive} $a^{\mathfrak B} \circ a^{\mathfrak B} = L_k$. This is a contradiction to Proposition \ref{prop: kein k2}.
\qed\end{proof}

\begin{lemma}\label{canonicalPolys}
	Let $a \in A$ be a symmetric atom of $\bf A$
	such that $\Aut({\mathfrak B})$ is primitive and $(a,a,a)$ is forbidden. Then all polymorphisms of $\mathfrak{B}$ are $\{\id, a\}$-canonical.
\end{lemma}

In the proof, we need the following notation. 
Let $ a_1, \ldots, a_k \in A$ be such that $a_1=\ldots =a_j$ and $a_{j+1}=\ldots =a_k$. Instead of writing $(a_1, \ldots , a_n)^\mathfrak{B}$ we use the shortcut $(a_1 |_j  a_{j+1})^\mathfrak{B}$. 

\begin{proof}[of Lemma~\ref{canonicalPolys}]
The following ternary relation $R$ on $B$ is primitive positive definable in $\mathfrak B$. 
	$$R := \big \{(x_1,x_2,x_3) \in B^3 \mid (a \cup \id)^\mathfrak{B}(x_1,x_2)  \wedge (a \cup \id)^\mathfrak{B}(x_2,x_3) \wedge a^\mathfrak{B}(x_1,x_3) \big \}$$
	Observe that $c \in R$ if and only if
	$a^\mathfrak{B}(c_1,c_2) \wedge \id^\mathfrak{B}(c_2,c_3)$ or $
	\id^\mathfrak{B}(c_1,c_2) \wedge a^\mathfrak{B}(c_2,c_3)$.
	
	Let $f$ be a polymorphism of $\mathfrak{B}$ of arity $n$. 
	Let $x,y,u,v \in B^n$ be arbitrary such that $(x,y)$ 
	and $(u,v)$ have the same $\{\id,a\}$-configuration. 
	Without loss of generality we may
	 assume that $(a |_j\id )^\mathfrak{B}(x,y)$ and 
	 $(a |_j\id)^\mathfrak{B}(u,v)$. Now consider $p,q \in B^n$ such that  $(\id |_j a)^\mathfrak{B}(p,q)$ holds.

	 Note that  by the edge-conservativeness of $f$ the following holds:
	  $$(f(x),f(y)) \in( a \cup \id)^\mathfrak{B},(f(u),f(v)) \in( a \cup \id)^\mathfrak{B}\textup{~and~} (f(p),f(q)) \in( a \cup \id)^\mathfrak{B}.$$

	By Lemma \ref{pathexistence} there exists a $k\in \mathbb{N}$ such that for every $i \in \{1,\dots,n\}$ there exists an $a^\mathfrak{B}$-walk $(s^0_i, \ldots, s^k_i)$ with $s^0_i = y_i$ 
	and $s^k_i =p_i$. 
	Now consider the following walk in $B^n$:
	\begin{align*}
	&(a |_j\id )^\mathfrak{B}(x,y) \\
	&(\id  |_j a )^\mathfrak{B}\big(y,(s^0_1, \ldots s^0_j, s^1_{j+1},\ldots s^1_{n}  )  \big) \\
	&( a  |_j \id )^\mathfrak{B}\big((s^0_1, \ldots s^0_j, s^1_{j+1},\ldots s^1_{n}  ), (s^1_1, \ldots s^1_j, s^1_{j+1},\ldots s^1_{n}  ) \big) \\
	& \quad \vdots \\
	&( a  |_j \id )^\mathfrak{B}((s^i_1, \ldots s^i_j, s^{i+1}_{j+1},\ldots s^{i+1}_{n}  ), (s^{i+1}_1, \ldots s^{i+1}_j, s^{i+1}_{j+1},\ldots s^{i+1}_{n}  )    ) \\
	&(  \id  |_j a)^\mathfrak{B}( (s^{i+1}_1, \ldots s^{i+1}_j, s^{i+1}_{j+1},\ldots s^{i+1}_{n}  ) , (s^{i+1}_1, \ldots s^{i+1}_j, s^{i+2}_{j+1},\ldots s^{i+2}_{n}  )  ) \\
	& \quad \vdots \\
	&(a  |_j \id )^\mathfrak{B}((s_1^{k-1},\dots,s^{k-1}_j,s_{j+1}^{k},\dots,s^{k}_n),p) \\
	&(\id  |_j a )^\mathfrak{B}(p,q) 
	\end{align*}
	
	Every three consecutive elements on this walk are component wise in the relation $R$. Since $R$ is primitive positive definable the polymorphism $f$ preserves $R$ by Theorem \ref{polinv}. This means that $f$ maps this walk on a walk where the atomic relations are an alternating sequence of $a^\mathfrak{B}$ and $\id^\mathfrak{B}$, which implies
	$$ (f(x), f(y)) \in a^\mathfrak{B} \Leftrightarrow (f(p), f(q)) \in \id^\mathfrak{B}.$$
	
	If we repeat the same argument with a walk from $q$ to $v$ we get:
		$$ (f(p), f(q)) \in a^\mathfrak{B} \Leftrightarrow (f(u), f(v)) \in \id^\mathfrak{B}.$$
		
	Combining these two equivalences gives us 
	
		$$ (f(x), f(y)) \in a^\mathfrak{B} \Leftrightarrow (f(u), f(v)) \in a^\mathfrak{B}.$$

Since the tuples $x,y,u,v \in B^n$ were arbitrary this shows that $f$ is $\{\id,a\}$-canonical. \qed\end{proof}

\begin{theorem}\label{Hensonalg}
	Let $\Aut(\mathfrak{B})$ be  primitive and let $a$ be a symmetric atom of $\mathbf{A}$ such that $(a,a,a)$ is forbidden. Then $\Csp(\mathfrak{B})$ is NP-hard.
\end{theorem}

\begin{proof}
	By Lemma \ref{canonicalPolys} we know that all polymorphisms of $\mathfrak{B}$ are $\{a,\id\}$-canonical. This means that every $f \in \pol(\mathfrak{B})$  induces an operation $\overline{f}$ of the same arity on the set $\{a,\id\}$.
	Let $\mathscr{C}_2$ be 
	 the set  of  induced operations. Note that $\mathscr{C}_2$ is an operation clone on a Boolean domain. 
	The mapping $\mu\colon \pol(\mathfrak{B} ) \rightarrow \mathscr{C}_2$ defined by $\mu(f):=\overline{f}$ is a uniformly continuous minor-preserving map.

	Assume for contradiction that there exists a ternary cyclic polymorphism $\overline{s}$ in $ \mathscr{C}_2 $. 
	Let $x,y,z \in B^3$ be such that

\begin{align*}
& (a,a,\id)^\mathfrak{B}( x,y), \\
& (\id,a,a)^\mathfrak{B}(y,z), \\
\textup{and~~}& (a,\id,a)^\mathfrak{B}(x,z).
\end{align*}

By the cyclicity of the operation $\overline{s}$  and the edge-conservativeness of $s$ we have that either
$$(s(x),s(y)) \in a^\mathfrak{B}, (s(y),s(z)) \in a^\mathfrak{B} \textup{~and~}   (s(x),s(z)) \in a^\mathfrak{B}$$
or
$$(s(x),s(y)) \in \id^\mathfrak{B}, (s(y),s(z)) \in \id^\mathfrak{B} \textup{~and~}   (s(x),s(z)) \in \id^\mathfrak{B}.$$
Since $(a,a,a)$ is forbidden, the second case holds.
Note that $\bf A$ must have an atom $b \neq \id$ such that the triple $(a,a,b)$ is allowed, because otherwise $a$ would be an equivalence relation.
Now consider $u,v,w \in B^3$ such that
\begin{align*}
& (a,a,\id)^\mathfrak{B}(u,v), \\
& (\id,a,a)^\mathfrak{B}(v,w), \\
\textup{and~~}& (a,b,a)^\mathfrak{B}(u,w).
\end{align*}
Since $s$ is $\{a,\id \}$-canonical and with the observation from before we have $$(s(u),s(v))\in \id^\mathfrak{B} \textup{~and~} (s(v),s(w))\in \id^\mathfrak{B}.$$
Now the transitivity  of equality contradicts $ (s(u),s(w)) \in (a\cup b)^\mathfrak{B}$.

We conclude that $\mathscr{C}_2$ does not contain a ternary cyclic operation. 
Since the domain of $ \mathscr{C}_2 $ has size two, Theorem \ref{existence cyclic} implies the existence of a u.c.\ minor-preserving map $\nu \colon \mathscr{C}_2  \rightarrow \proj$. The composition $\nu \circ \mu \colon \pol(\mathfrak{B}) \rightarrow \proj$ is also a u.c.\ minor-preserving map and therefore by Theorem \ref{hardness infinite ucmp} the $\Csp(\mathfrak{B})$ is NP-hard. \qed\end{proof}

\section{Examples}
\label{sect:expl}

\begin{figure}[t]
	\begin{center}
		
		\begin{minipage}[b]{40mm}
			\begin{center}
				
				\begin{tabular}{|c||c|c|c|}
					\hline 
					$~\circ~$	& $~\id~$ & $~a~$ & $~b~$ \\ 
					\hline \hline
					$\id$	&  $\id$ & $a$ &$b$  \\ 
					\hline 
					$a$	&$a$  & $\neg b$  & $b$ \\ 
					\hline 
					$b$	& $b$ & $b$ & $\neg b$ \\ 
					\hline 
				\end{tabular} 
				
			\end{center}
		\end{minipage}
		\begin{minipage}[b]{40mm}
			\begin{center}
				
				\begin{tabular}{|c||c|c|c|}
					\hline 
					$~\circ~$	& $~\id~$ & $~a~$ & $~b~$ \\ 
					\hline \hline
					$\id$	&  $\id$ & $a$ &$b$  \\ 
					\hline 
					$a$	&$a$  & $\neg a$  & $0'$ \\ 
					\hline 
					$b$	& $b$ & $0'$ & $1$ \\ 
					\hline 
				\end{tabular} 
			\end{center}
		\end{minipage}
	\end{center}
	\caption{Multiplication tables of relation algebras $\#13$ (left) and $\#17$ (right).}
	\label{Tables}
\end{figure}

Andr\'eka  and Maddux classified \textit{small relation algebras}, i.e., finite relation algebras with at most 3 atoms \cite{AndrekaMaddux}. We consider the complexity of the network satisfaction problem of two of them, namely the relation algebras $\#13$ and $\#17$ (we use the enumeration from \cite{AndrekaMaddux}). Both relation algebras have normal representations (see below) and fall into the scope of our hardness criteria. Cristani and Hirsch \cite{HirschCristiani} classified the complexities of the network satisfaction problems for small relation algebras, but due to a mistake the algebras $\#13$ and $\#17$ were left open. 

\begin{example}[Relation Algebra $\#13$]
		The relation algebra $\#13$ is given by the multiplication table in Fig. \ref{Tables}. This finite relation algebra has a normal representation $\mathfrak{B}$ defined as follows. Let $V_1$ and $V_2$ be countable, disjoint sets. We set $B:= V_1 \cup V_2$ and define the following atomic relations: 
	\begin{align*}
	&\id^\mathfrak{B} := \{(x,x) \in B^2 \}, \\
	& a^\mathfrak{B} := \{(x,y) \in B^2\setminus \id^\mathfrak{B}  \mid (x\in V_1 \wedge y\in V_1) \vee (x\in V_2\wedge y \in V_2) \},\\
		& b^\mathfrak{B} := \{(x,y) \in B^2 \setminus \id^\mathfrak{B} \mid (x\in V_1 \wedge y\in V_2) \vee (x\in V_2\wedge y \in V_1) \}.
	\end{align*}
	It is easy to check that this structure is a square representation for $\#13$. Moreover, this structure is fully universal for $\#13$ and homogeneous, and therefore a normal representation. 
	
	Note that the relation $(\id \cup~a)^\mathfrak{B}$ is an equivalence relation where $V_1$ and $V_2$ are the two equivalence classes. Therefore we get by Theorem \ref{hardness:finitely many classes} that the (general) network satisfaction problem for the relation algebra $\#13$ is NP-hard.
	We mention that this result can also be deduced from the results in~\cite{BMPP16}. 
	\end{example}

\begin{example}[Relation Algebra $\#17$]
	The relation algebra $\#17$ is given by the multiplication table in Fig. \ref{Tables}. Let $\mathfrak{N}=(V; E^\mathfrak{N})$  be the countable, homogeneous, universal triangle-free, undirected graph (see~\cite{Hodges}), also called  called a \emph{Henson graph}.
	We use this Henson graph to obtain a square representation $\mathfrak{B}$ with domain $V$ for the relation algebra $\#17$ as follows:
	\begin{align*}
	&\id^\mathfrak{B} := \{(x,x) \in V^2 \}, \\
& a^\mathfrak{B} := \{(x,y) \in V^2  \mid (x,y)\in E^\mathfrak{N} \},\\
& b^\mathfrak{B} := \{(x,y) \in B^2 \setminus \id^\mathfrak{B} \mid (x,y)\not \in E^\mathfrak{N} \}.
	\end{align*}
	This structure is homogeneous and fully universal since $\mathfrak{N}$ is homogeneous and embeds every triangle free graph. 
	 It is easy to see that there exists no non-trivial equivalence relation in this relation algebra. For the atom $a$ the triangle $(a,a,a)$ is forbidden, which means we can apply Theorem \ref{Hensonalg} and get NP-hardness for the (general) network satisfaction problem for the relation algebra $\#17$. 	Also in this case,
	 the hardness result can also be deduced from the results in~\cite{BMPP16}. 
\end{example}

\section{Conclusion and Future Work}
To the best of our knowledge the computational complexity of the (general) network satisfaction problem was previously only known for a small number of isolated finite relation algebras, for example the point algebra, Allens interval algebra,  or the 18 small relation algebras from~\cite{AndrekaMaddux}. Both of our criteria, Theorem \ref{hardness:finitely many classes} and Theorem \ref{Hensonalg}, show the NP-hardness for relatively large classes of finite relation algebras. In Section~\ref{sect:expl} we applied
these results to settle the
complexity status of two problems that were left open in~\cite{HirschCristiani}. 

To obtain our general hardness conditions we used the universal algebraic approach for studying the complexity of constraint satisfaction problems. This approach will hopefully lead to a solution of Hirsch's RBCP for
all finite relation algebras $\bf A$ with a normal representation $\mathfrak B$. It is also relatively easy to prove
that the network satisfaction problem for $\bf A$
is NP-complete if $\mathfrak B$ has an equivalence relation with an equivalence class of finite size larger than two.  
Hence, the next steps that have to be taken with this approach are the following. 
\begin{itemize}
\item Classify the complexity of the network satisfaction problem for finite relation algebras
$\bf A$ where the normal representation has a primitive automorphism group.
\item Classify the complexity of the network satisfaction problem for relation algebras
that have equivalence relations with infinitely many
classes of size two.
\item Classify the complexity of the network satisfaction problem for relation algebras
that have equivalence relations with infinitely many infinite
classes. 
\end{itemize}

\bibliographystyle{alpha}

\bibliography{local.bib}

\end{document}